\documentclass[a4paper,12pt,twoside,onecolumn]{article}
\usepackage{amsmath,amsthm,amsfonts,latexsym,amscd,amssymb, graphicx, color,enumerate, caption}

\usepackage[hidelinks]{hyperref}
\newtheorem{thm}{Theorem}[section]
\newtheorem{dfn}[thm]{Definition}
\newtheorem{rmk}[thm]{Remark}
\newtheorem{exl}[thm]{Example}

\newtheorem{propo}[thm]{Proposition}
\newtheorem{coro}[thm]{Corollary}

\newtheorem{lem}[thm]{Lemma}

\begin{document}
	\title{\bf { $W-\delta-\mu$ dual codes and LCD codes} }  
	\author{\bf Avanish Kumar Chaturvedi$^{a}$\footnote{Corresponding author} and \bf Satyadeep Pandey$^{b}$} 
	\date{Department of Mathematics, University of Allahabad \\ Prayagraj-211002, India \\ {\footnotesize  $^{a}$akchaturvedi.math@gmail.com, achaturvedi@allduniv.ac.in\\
			$^{b}$ pandeysatyam14082000@gmail.com}}     	
	\maketitle   
	\begin{abstract}
		We introduce a new product on the ambient space $F_q^n$ as a generalization  of Euclidean, Hermitian and $\delta$ products. We give some general properties of the dual codes, relation with Euclidean duals, definition and characterization of self orthogonal, self dual, dual containing and LCD codes along with certain existence conditions. Also, we calculate the dual codes of some classes of codes like repetition, binary and $\lambda$-constacyclic codes with respect to this product. Further, we extend and analyse this notion of the product for codes over semisimple rings.
		\end{abstract} 
		

		
	\section{Introduction} In coding theory, the notion of linear codes is of paramount importance for their easy implementation and handling in terms of generator and parity check matrices. When talking about linear codes, a great deal of significance is attached to their duals and the resulting interrelationships. As the research advances, focus has been on using different products for getting dual codes with desired properties for particular purposes.
	
	In this line of research, Fan and Zhang in 2017 \cite{Fan}, introduced the $l-Galois\: inner\: product$ generalizing the Euclidean and Hermitian product. This was further generalized by Dinh et al. in 2024 \cite{Dinh}, where they introduce $\delta$-product in which the second coordinate contributes $\delta$ twisted values in the product, where $\delta$ is a field automorphism.
	
	Szabo et al. studied non degenerate forms on the ambient space and the duals of codes with respect to these forms, exploring properties such as Mac williams relations among others (see \cite{Szabo}). The attempts have been to tweak the usual product in such a way that the resulting form is concrete and has manageable parameters for applications. In \cite{Meng} and \cite{Claude}, $\sigma$- inner product and $\sigma$-duals have been studied for classical and quantum codes where $\sigma$ is a semilinear automorphism on the ambient space $F_q^n$.
	
	The relation of a linear code with its dual can vary from self orthogonality to linear complementarity. It is also an area of interest to search for such codes, as they can have utility in new constructions such as quantum code constructions. In \cite{Egan}, the weighing matrix concept is used to construct self orthogonal codes.
	
	Such constructions are not limited for finite field alphabets only but also applicable for various classes of rings. For example in  \cite{Dinh} authors introduced $\delta$ inner product over semisimple rings and explored duals of linear codes over semisimple rings.
	
	Motivated by above studies we define a new product on vector spaces over a finite field having parameters $W$, $\delta$ and $\mu$, which offers flexibility in terms of dual constructions by virtue of scope of variation in these parameters. We explore properties of linear codes with respect to this product. Moreover, the definition of the product is generalized to the case where the code alphabet is a finite semi-simple ring $R$ and the codes are submodules of the free module $R^n$. We obtain various results regarding duals over $R$  using the results that we establish over a finite field alphabet. 
	
	The computational results presented in this work demonstrate that the proposed $W$--$\delta$--$\mu$ product offers significantly greater flexibility than the $\delta$-product in the construction of dual codes. For a fixed generator matrix over $\mathbb{F}_4$, the number of admissible products is significantly enlarged by allowing independent choices of $W$, field automorphisms $\delta$, and permutations $\mu$. Consequently, a single linear code can admit numerous distinct dual structures, thereby considerably expanding the range of attainable $W$--$\delta$--$\mu$ LCD and $W$--$\delta$--$\mu$ self-orthogonal codes. However, in support, we show in Example \ref{Better example} that since the $\delta$-automorphism is restricted to only two possible choices corresponding to the identity and Frobenius automorphisms, ther are only two choices of $\delta$-product. The computational results clearly show that the proposed product provides a much richer framework for investigating generalized duality.
	
	Another noteworthy observation is the agreement between the theoretical developments and the computational results. The example (see Example \ref{Better example}) shows that whenever the hypotheses of Proposition~\ref{Proposition:4.9} are satisfied, $W$--$\delta$--$\mu$ self-orthogonal codes surely exist. Moreover, for several generator matrices, the proposed product produces multiple choices of $W$, $\delta$ and $\mu$ yielding \emph{$W$--$\delta$--$\mu$ LCD codes} as well as numerous \emph{$W$--$\delta$--$\mu$ self-orthogonal codes}. These computations therefore not only validate the theoretical results but also demonstrate that the additional parameters introduced through $W$ and $\mu$ have practical significance in identifying new code families possessing prescribed duality properties.
	
	The outline of the paper is as follows. In section 2, we discuss some preliminaries important for the development of the rest of the sections. In section 3, we define the $W-\delta-\mu$ product establishing linearity in first, $\delta$ semilinearity in second and non degeneracy in both coordinates [\ref{4.1}]. We show that many existing products become particular cases of the $W-\delta-\mu$ product. An example [\ref{product example}] illustrating the working of the product has been given which also shows that it is non symmetric. Further we also give matrix representation of the product [\ref{3.5}].
	
	In section 4, we define the dual of a linear code with respect to the $W-\delta-\mu$ product. We establish that the $W-\delta-\mu$ dual is equal to the inverse image of the Euclidean dual under a $\delta$ semilinear bijection on $F_q^n$ [\ref{4.3}]. We show that the Mac Williams relations hold for the $W-\delta-\mu$ product [\ref{Mac Williams}]. Further we establish that for the notions of Euclidean and $W-\delta-\mu$ duals to coincide under a suitable condition on $\delta$ the choices for $W$ and $\mu$ are highly restricted [\ref{coincide theorem}]. We illustrate the dual calculation with the help of an example [\ref{dual example}]. Moreover, in subsection 4.1, we characterize the $W-\delta-\mu$ self orthogonality in terms of generator matrices [\ref{dual matrix}]. We also show [\ref{Proposition:4.9}] that for suitable $q$ and $n$ there always exists a $W$ which makes a given linear code $C$, $W-\delta-\mu$ self orthogonal.
	
	In section 5, we discuss $W-\delta-\mu$ dual containing and LCD codes. We give examples to show how these are different from Euclidean dual containing and LCD codes. We give a necessary and sufficient condition for existence of $W$ which makes a linear code $C$ dual containing [\ref{5.8}]. Further we characterize the LCD codes in terms of generator matrices [\ref{LCD matrix}]. Using Schwartz Zippel lemma we establish that for large enough code alphabets there always exists a $W$ which makes a given linear code $C$, $W-\delta-\mu$ LCD for arbitrary $\delta$ and $\mu$ [\ref{5.12}].
	
	In section 6, we calculate the $W-\delta-\mu$ duals of some classes of codes such as repetition, binary and $\lambda$ constacyclic codes. We give an example showing that the $W-\delta-\mu$ dual of a $\lambda$ constacylic code need not be constacylic [\ref{consta example}].
	
	In section 7, we generalize the definition of $W-\delta-\mu$ product to the case when the code alphabet is a finite semisimple ring. We obtain the component wise decomposition of a linear code [\ref{decomposition}] and its $W-\delta-\mu$ dual [\ref{8.2}] over the constituent fields. We establish the existence results for $W$ for a linear code to be $W-\delta-\mu$ self orthogonal and LCD over a semisimple ring.
	
	Finally, in section 8 we discuss conclusion of the work and describe future scope along with potential applications of the $W-\delta-\mu$ product.

	\section{Preliminaries} 
	
	Following \cite{Xing}, by a \emph{linear code} $C$ of length $n$ over alphabet $F_q$, a finite field, we mean an $F_q$ subspace of the vector space $F_q^n$. In case of a finite ring $R$, a linear code over $R$ of length $n$ is an $R$-submodule of $R^n$.
	
	If $C$ is a linear code over a finite field $F_q$ of dimension $k$ then it has a basis, say, $\{g_1,g_2,...,g_k\}\subseteq F_q^n$. A \emph{generator matrix} for $C$ is  $G=[g_{ij}]_{k\times n}$, where, $g_i=(g_{ij})_{j=1}^n\in F_q^n$. Then
	$C=\{uG\;:\;u\in\;F_q^k\}.$
	
	The \emph{Euclidean product} or the \emph{usual product} of two codewords $x=(x_1,x_2,...,x_n)$ and $y=(y_1,y_2,...,y_n)$ in $F_q^n$ is given by  $<x,y>_{Euclidean}=\displaystyle\sum_{i=1}^{n}x_iy_i$. With respect to the Euclidean product, we have the standard notion of dual space. Referring to \cite{Xing}, the \emph{Euclidean dual} of a linear code $C$ is
	$C^{\perp_{Euclidean}}=\{y\in F_q^n\;:\;<x,y>_{Euclidean}=0,\;\forall\;x\in\:C\}.$
	
	It is a vector subspace of $F_q^n$ with dimension $n-k$, where $k$ is the dimension of $C$ over $F_q$. A generator matrix $H$ for the Euclidean dual of a linear code $C$ is called a \emph{parity check matrix} for $C$ . We have the following relation between generator and parity check matrices.
	\begin{lem}\cite[Remark 4.5.5]{Xing}\label{GH theorem}
		Let $C$ be an $[n,k]$-linear code over $Fq$, with parity-check matrix $H$. Then $v\in F_q^n$ belongs to $C$ if and only if $v$ is orthogonal to every row of $H$; i.e.,
		$v\in C\iff vH^T = 0$. In particular, a $k \times n$ matrix $G$, is a
		generator matrix for $C$ if and only if the rows of $G$ are linearly independent
		and $GH^T=0$.
		
	\end{lem}
		Following \cite{Alqarni}, we write, an additive map $\Gamma:R\longrightarrow R$ is referred to as an \emph{antiautomorphism} if it is bijective and satisfes $\Gamma(b)\Gamma(a)=\Gamma(ab)$ for every $a,b\in R.$ In the following, we recall \emph{$\sigma-$ sesquilinear form} as a generalization of bilinear form. 
	\begin{dfn}\cite[Section 3]{Fan}\label{2.1}
		For a ring A, an anti-automorphism $\sigma$ on A and a left A-module M, a $\sigma$-sesquilinear form on
		M is a map $<,>:M \times M \longrightarrow A $ such that if $x,y,z \:\in\: M$ and $a \:\in\: A$, then $<x+z,y>
		= <x,y>+<z,y>,\;<x+y,z>=<x,z>+<y,z>,\; <ax,y>=
		a<x,y>$ and $<x,ay>=<x,y>\sigma(a)$. 
	\end{dfn}
	
	The error detection and correction capability of a code is inherently linked to the minimum Hamming distance parameter which is equivalently described by the notion of Hamming weight. Following \cite{MacWilliams}, the \emph{Hamming weight} \(w(c)\) of a codeword \(c\) within a code is defined as the number of non-zero alphabets in \(c\). Also, for a code $C$ of length $n$, the \emph{Hamming weight enumerator polynomial} is a bivariate polynomial given by \(W_{C}(x,y)=\sum _{i=0}^{n}A_{i}x^{n-i}y^i\), where \(A_{i}\) is the number of codewords with weight \(i\).
	
	For a given linear code $C$ over $F_q$ and its dual $C^{\perp}$, with respect to any product, we can study various special cases such as $C\subseteq C^{\perp}$ (\emph{self orthogonal code}), $C= C^{\perp}$ (\emph{self dual code}), $ C^{\perp}\subseteq C$ (\emph{self dual containing}) and $C\cap C^{\perp}=\{0\}$ (\emph{linear complementary dual}, i.e., LCD codes).
	
	Within the class of linear codes, the \emph{cyclic codes} preserving cyclic shift property, are of particular importance for their easy description and implementation. Readers are referred to \cite[Chapter 4]{Huffman} for a detailed discussion on cyclic codes. Cyclic codes have been generalized to the class of \emph{constacyclic codes} which offer a wider class of algebraic codes. Following \cite{Minjia}, a linear code \(C\) of length \(n\) over a finite field \(F_q\) is called a \emph{\(\lambda \)-constacyclic code} (for some $
	\lambda\in F_q$) if it is closed under the \(\lambda \)-constacyclic shift operation \(T_{\lambda }\), defined as: \(T_{\lambda }(c_{0},c_{1},\dots ,c_{n-1})=(\lambda c_{n-1},c_{0},\dots ,c_{n-2})\) for every codeword \(c=(c_{0},c_{1},\dots ,c_{n-1})\in C\).
	
	A constacyclic code over a finite field $F_q$ is characterised  as an ideal of a quotient ring of the polynomial ring over $F_q$.
	\begin{propo}\cite[proposition 2.2]{Dinh}
		A linear code \(C\) of length $n$ is \(\lambda \)-constacyclic over $F$ if and only if \(C\) is an ideal of \(F_[x]/\langle x^{n}-\lambda \rangle \).
	\end{propo}
	A \emph{semisimple ring} is a ring that is semisimple as a module over itself, meaning it can be expressed as a finite direct sum of simple left (or right) ideals. In case of commutativity, it is well known that a commutative semisimple ring $R$ is a ring that is isomorphic to a finite direct product of fields, i.e, $R=\displaystyle\bigoplus_{i=1}^{n}F_i$, for some $n\in \mathbb{N}$, where $F_i$ is a field [cf \cite{Donald}].
	
	%%%%%%%%%%%%%%%%%%%%%%%%%%%%%%%%%%%%%%%%%%%%%%%%
	\section{$W$-$\delta$-$\mu$ product}
	\begin{dfn}	Let $F_q$ be a finite field, $ W = \{w_1, w_2, ..., w_n\} \subseteq F_q^*, \delta$ be a field automorphismm on $F_q$ and $\mu \in S_n$ be a permutation on  $ n $ symbols. We call, a map $<,>_{W-\delta-\mu} : F_q^n \times F_q^n \rightarrow F_q$ defined by
		$<x,y>_{W-\delta-\mu}\; = \displaystyle\sum_{i=1}^{n}w_ix_i\delta(y_{\mu(i)})$,
		where $x =(x_1,x_2,...,x_n), y=(y_1,y_2,...,y_n) \in F_q^n$ a \textbf{\emph{$W$-$\delta$-$\mu$ product}} on $F_q^n$.
		
		If $\delta=I_d$ and $\mu=I_d$ then we call the $W-\delta-\mu$ product as \emph{Weighted product}.
		
	    If $W=\{1\}$ and $\delta=I_d$ then we call the $W-\delta-\mu$ product as \emph{Permutation product}.
		
		\end{dfn}
	We first show that the $W-\delta-\mu$ product satisfies some useful basic properties of being a product.
	\begin{lem}\label{4.1}
		The ${W-\delta-\mu}$ product is linear in first, $\delta$-semilinear in second and non degenerate in both coordinates.
	\end{lem}
	\begin{proof}Let $a,b \in F_q, x,y,z \in F_q^n$. First, we show linearity in the first co-ordinate. We have,
		\[
		\begin{aligned} <ax+by,z>_{W-\delta-\mu}\;&=\displaystyle \sum_{i=1}^{n}w_i(ax_i+by_i)\delta(z_{\mu(i)})
			\\&=a\displaystyle\sum_{i=1}^{n}w_ix_i\delta(z_{\mu(i)}) + b\displaystyle\sum_{i=1}^{n}w_iy_i\delta(z_{\mu(i)})
			\\&=a<x,z>_{W-\delta-\mu}+b<y,z>_{W-\delta-\mu}\end{aligned}
		\]
		
		Now, we show $\delta$-semilinearity in the second co-ordinate.
		We have,
		\[
		\begin{aligned}
			<x,a(y+z)>_{W-\delta-\mu}\;&= \displaystyle\sum_{i=1}^{n}w_ix_i\delta((ay+bz)_{\mu(i)})\\
			&=\displaystyle\sum_{i=1}^{n}w_ix_i\delta(ay_{\mu(i)}+bz_{\mu(i)})\\
			&=\displaystyle\sum_{i=1}^{n}w_ix_i\delta(ay_{
				\mu(i)}) + \displaystyle\sum_{i=1}^{n}w_ix_i\delta(bz_{\mu(i)})\\
			&=\delta(a)<x,y>_{W-\delta-
				\mu} + \delta(b)<x,z>_{W-\delta-\mu}
		\end{aligned}
		\]
		
		Finally, we show non degeneracy, in first co-ordinate. Let $x \in F_q^n$ such that $<x,y>_{W-\delta-\mu}=0,$ $ \forall y \in F_q^n$. Then
		\[
		\displaystyle\sum_{i=1}^{n}w_ix_i\delta(y_{\mu(i)})=0, \forall y \in F_q^n. \tag{1}
		\]
		For an arbitrary $i$, consider $y=(y_j)_n \in F_q^n $ such that $ y_j=1$ if $j=\mu(i)$ and $0$ otherwise. Then from (1), $w_ix_i=0$ which implies $x_i=0$ since $w_i$ is non zero. Hence, $x=0$.
		
		Next, for non degeneracy in second co-ordinate, let $y \in F_q^n$ such that $<x,y>_{W-\delta-\mu}=0, \forall x \in F_q^n$.	Then
		\[
		\displaystyle\sum_{i=1}^{n}w_ix_i\delta(y_{\mu(i)})=0, \forall x=(x_i)_1^n \in F_q^n \tag{2}.
		\]
		For an arbitrary $i$, we choose $x$ such that $x_j=1$ if $j=i$ and 0 otherwise. Then, from (2), $$w_ix_i\delta(y_{
			\mu(i)})=0
		\Rightarrow  \delta(y_{
			\mu(i)})=0 \Rightarrow y_{\mu(i)}=0 \Rightarrow y=0.$$
		Hence,$<,>$ is non degenerate in the second coordinate.
	\end{proof}
	\par 
	We note that this product generalizes Euclidean product if $ W=\{1\}$,
	$\delta$ is the identity on $F_q$ and $\mu$ is the identity permutation. If we consider $W=\{1\}$,
	$\delta$ is the Frobenius automorphism on $ F_q$ and $\mu$ is the identity permutation then it generalizes Hermitian product. Also it generalizes to $\delta$ product if $W=\{1\}, \delta$ is an $F_q$ automorphism and $\mu$ is the identity permutation.
		\begin{table}[ht]
		\centering
		\caption{Special cases of the $W$-$\delta$-$\mu$ product}
		\renewcommand{\arraystretch}{1.2}
		\begin{tabular}{|c|c|c|c|}
			\hline
			\emph{Product} & \emph{$W\subseteq F_q^{*}$} & \emph{$\delta$} & \emph{$\mu$}\\
			\hline
			Euclidean
			&
			$\{1\}$
			&
			$\mathrm{id}$
			&
			$\mathrm{id}$
			\\
			\hline
			Hermitian
			&
			$\{1\}$
			&
			Frobenius automorphism
			&
			$\mathrm{id}$
			\\
			\hline
			$\delta$-product
			&
			$\{1\}$
			&
			Any automorphism
			&
			$\mathrm{id}$
			\\
			\hline
			Weighted product
			&
			W
			&
			$\mathrm{id}$
			&
			$\mathrm{id}$
			\\
			\hline
			Permutation product
			&
			$\{1\}$
			&
			$\mathrm{id}$
			&
			Any permutation
			\\
			\hline
			$W$-$\delta$-$\mu$
			&
			W
			&
			Any automorphism
			&
			Any permutation
			\\
			\hline
		\end{tabular}
	\end{table}
	\begin{exl}
	
	Let ${F}_4=\{0,1,\alpha,\alpha^2\}, \alpha^2=\alpha+1,$
	and let $x=(1,\alpha,\alpha^2),
	y=(\alpha^2,\alpha,1).$	
	The nontrivial automorphism is $\delta(a)=a^2,$
	so that $\delta(0)=0,
	\delta(1)=1,
	\delta(\alpha)=\alpha^2,
	\delta(\alpha^2)=\alpha.$
	\begin{table}[h]
		\centering
		\captionsetup{justification=centering}
			\caption{Comparison of various products and their corresponding values on the vectors
			$x=(1,\alpha,\alpha^2)$ and
			$y=(\alpha^2,\alpha,1)$ over $ F_4$.}
		\renewcommand{\arraystretch}{1.3}
		\begin{tabular}{|l|l|l|l|c|}
			\hline
			\emph{Product} & \emph{$W\subseteq F_4^{*}$} & \emph{$\delta$} & \emph{$\mu$} & \emph{Value}\\
			\hline
			Euclidean
			&
			$\{1\}$
			&
			$\mathrm{id}$
			&
			$\mathrm{id}$
			&
			$\alpha^2$
			\\
			\hline
			Hermitian
			&
			$\{1\}$
			&
			$a\mapsto a^2$
			&
			$\mathrm{id}$
			&
			$0$
			\\
			\hline
			$\delta$-product
			&
			$\{1\}$
			&
			$a\mapsto a^2$
			&
			$\mathrm{id}$
			&
			$0$
			\\
			\hline
			Weighted
			&
			$\{1,\alpha,\alpha^2\}$
			&
			$\mathrm{id}$
			&
			$\mathrm{id}$
			&
			$0$
			\\
			\hline
			Permutation
			&
			$\{1\}$
			&
			$\mathrm{id}$
			&
			$(123)$
			&
			$\alpha$
			\\
			\hline
			$W$-$\delta$-$\mu$
			&
			$\{1,\alpha,\alpha^2\}$
			&
			$a\mapsto a^2$
			&
			$(123)$
			&
			$\alpha^2$
			\\
			\hline
		\end{tabular}
	\end{table}
\end{exl}
	In the following, we give an example for the further illustration of this product.
	
	\begin{exl}\label{product example}
		Consider $F_4 =\{0,1,\alpha,\alpha^2\}$, where $\alpha^2=\alpha+1$. Let $W= \{1,\alpha,1\}, \delta(x)=x^2$ and $\mu=(2\:3)\in S_3$. For $u=(u_1,u_2,u_3), v=(v_1,v_2,v_3)$ in $F_4^3$ we define, $<,>_{W-\delta-\mu}: F_4^3 \times F_4^3 \rightarrow F_4$ by $<u,v>_{W-\delta-\mu} = 1u_1\delta(v_{\mu(1)}) + \alpha u_2\delta(v_{\mu(2)})+1u_3\delta(v_{\mu(3)})
		=\:u_1\delta(v_1)\:+\alpha u_2\delta(v_3)+u_3\delta(v_2)
		=\:u_1v_1^2\:+ \alpha u_2v_3^2+u_3v_2^2.$
		For instance, if we take $u =\: (1,0,1)$ and $v=(1,1,0)$ then $<u,v>_{W-\delta-\mu}=0$ but $<v,u>_{W-\delta-\mu}=\:\alpha +1$.
	\end{exl}
	
	Thus, we observe that $v$ is orthogonal to $u$ but $u$ is not orthogonal to $v$. Also, it shows that the product is non symmetric.
	
	\begin{propo}
		
		The $W-\delta-\mu$ product induces a $\delta-$ sesquilinear form on $F_q^n$.\end{propo}
	
	\begin{proof} In Definition \ref{2.1}, if we take $A=F_q$ a finite field and $M=F_q^n$, then anti automorphism $\sigma$ on $A$ is equal to automorphism $\delta$ on $F_q$. By Lemma \ref{4.1} and Definition \ref{2.1}, it follows that $W-\delta-\mu$ product is a $\delta$ sesquilinear form. 
	\end{proof}
	\begin{rmk}\label{3.5}
		We can write the $W-\delta-\mu$ product in terms of a suitable product of matrices. Let $u =(u_1,\:u_2,\:\cdots\:,u_n)$, $v =(v_1,\:,v_2,\:\cdots\:,v_n)\in F_q^n$ and $\{e_1,e_2,...,e_n\}$ be the usual basis of $F_q^n$. Then $<u,v>_{W-\delta-\mu}$ can be represented by the following matrix product:
		
		\[
		\begin{bmatrix}
			u_1 & u_2 & \cdots & u_n\\
		\end{bmatrix} 
		\;
		\begin{bmatrix}
			w_1 & 0 & 0 & \cdots & 0\\
			0 & w_2 & 0 & \cdots & 0\\
			\vdots & \vdots & \vdots & \ddots & \vdots\\
			0 & 0 & 0 & \cdots & w_n
		\end{bmatrix}
		\;
		\begin{bmatrix}
			e_{\mu(1)} \\
			e_{\mu(2)} \\
			\vdots \\
			e_{\mu(n)}
		\end{bmatrix}
		\;
		\begin{bmatrix}
			\delta(v_1)\\
			\delta(v_2)\\
			\vdots\\
			\delta(v_n)
		\end{bmatrix},
		\] where $e_{\mu(1)}, e_{\mu(2)}, \ldots, e_{\mu(n)}\in\:F_q^n.$ 
		If we denote, $u=[u_1\:u_2\:...u_n]$, $D_W= diag(w_1,w_2,...,w_n)$, $P=\begin{bmatrix}
			e_{\mu(1)} \\
			e_{\mu(2)} \\
			\vdots \\
			e_{\mu(n)}
		\end{bmatrix}$, $\Delta(v)=\begin{bmatrix}
			\delta(v_1) &
			\delta(v_2)&
			\hdots&
			\delta(v_n)
		\end{bmatrix}$, then $<u,v>_{W-\delta-\mu}= u D_WP[\Delta(v)]^T.$
	\end{rmk}
	
	\section{$W$-$\delta$-$\mu$ dual codes}
	\begin{dfn} Let $C$ be a code over $F_q$. Then the $W-\delta-\mu$ dual of $C$, denoted by $C^{\perp_{W-\delta-\mu}}$ is defined as;
		$C^{\perp_{W-\delta-\mu}}\;=\{y\in F_q^n\;:\; <x,y>_{W-\delta-\mu}=0\; \forall \; x \in C\}.$
	\end{dfn}

	\begin{propo} If $C$ is a linear code over $F_q$, then $C^{\perp_{W-\delta-\mu}}$ is again linear.
	\end{propo}
	
	\begin{proof} Let $s_1, \;s_2 \in\; F_q,\; y,z \in C^{\perp_{W-\delta-\mu}}$. Then for every $x$ in $C$, $<x,y>_{W-\delta-\mu}=0=<x,z>_{W-\delta-\mu}$. By Lemma \ref{4.1}, it follows that
		$<x,s_1y+s_2z>_{W-\delta-\mu}=\delta(s_1)<x,y>_{W-\delta-\mu}+\delta(s_2)<x,z>_{W-\delta-\mu}=0$.
		Hence, $s_1y+s_2z \in C^{\perp_{W-\delta-\mu}}$. Thus, $C^{\perp_{W-\delta-\mu}}$ is linear over $F_q$.
	\end{proof}
	Now, we show that the sum of dimensions of a linear code $C$ over $F_q$ and its $W-\delta-\mu$ dual is equal the length of the code. 
	First, we give the following result.
	\begin{lem}\label{4.3}
		Let $C$ be a linear code over $F_q$. Then there exists a $\delta$-semilinear bijection $T$ on $F_q^n$ such that $C^{\perp_{W-\delta-\mu}}=T^{-1}(C^{\perp_{Euclidean}})$.
	\end{lem}
	\begin{proof}
		We define, $T:F_q^n\longrightarrow F_q^n$ such that for every $i$, $(T(v))_i=w_i\delta(v_{\mu(i)})$, $v\in F_q^n,\:w_i\in W$. For any $v\in F_q^n$, consider $u\in F_q^n$ with $u_i=\delta^{-1}(w^{-1}_{\mu^{-1}(i)}v_{\mu^{-1}(i)})$. Then $T(u)_i=v_i$. Thus, $T(u)=v$. Hence, $T$ is onto. It is easy to check that $T$ is one-one and $\delta-$ semilinear.
		Now,
		\[\begin{aligned}
			v\in C^{\perp_{W-\delta-\mu}} &\iff <u,v>_{W-\delta-\mu}=0\\
			&\iff <u,T(v)>_{Euclidean}=0\\
			&\iff T(v)\in\;C^{\perp_{Euclidean}}\\
			&\iff v\in\; T^{-1}(C^{\perp_{Euclidean}}).
		\end{aligned}
		\] 
		Hence, $C^{\perp_{W-\delta-\mu}}=T^{-1}(C^{\perp_{Euclidean}})$.
	\end{proof}
	
	\begin{propo}\label{4.4} Let $C$ be a linear code of length $n$ over $F_q.$ Then 
		\begin{enumerate}[(1)]
			\item $dim(C)+dim(C^{\perp_{W-\delta-\mu}})=n$ and
			\item $|C||C^{\perp_{W-\delta-\mu}}|=q^n.$			
		\end{enumerate}
	\end{propo}
	\begin{proof} It follows by Lemma \ref{4.3} that both Euclidean and $W-\delta-\mu$ duals have the same cardinality. Thus,
		$dim(C^{\perp_{Euclidean}})=n-dim(C)=dim(C^{\perp_{W-\delta-\mu}})$. Hence, $dim(C)+dim(C^{\perp_{W-\delta-\mu}})=n$. The second result directly follows from (1).
	\end{proof}
	
	In \cite[Chapter 5, Theorem 13]{MacWilliams} authors discuss Mac Williams theorem for Hamming weight enumerators, given by, \(W_{C^{\perp }}(x,y)=\frac{1}{|C|}W_{C}(x+(q-1)y,x-y)\), where, $C$ is a linear code over a finite field $F_q$ with \(W_{C}(x,y)\) as its Hamming weight enumerator polynomial, \(C^{\perp }\) as its Euclidean dual, and \(W_{C^{\perp }}(x,y)\) as the weight enumerator polynomial of \(C^{\perp }\).
	In the following we show that Mac Williams relations are also valid for the $W-\delta-\mu$ product.
	
	\begin{propo}\label{Mac Williams} Let $C$ be a linear code over $F_q$. Then the Hamming weight enumerator polynomials of $C$ and $C^{\perp_{W-\delta-\mu}}$ satisfy the Mac Williams relations.\end{propo}
	
	\begin{proof} Following the proof of Lemma \ref{4.3}, the map $T:F_q^n\longrightarrow F_q^n$ such that for every $i$, $(T(v))_i=w_i\delta(v_{\mu(i)})$, $v\in F_q^n,\:w_i\in W$ is bijective, $<u,v>_{W-\delta-\mu}=<u,T(v)>_{Euclidean} $ and $C^{\perp_{W-\delta-\mu}}=T^{-1}(C^{\perp_{Euclidean}}).$ Also, if $w_H$ denotes the Hamming weight, then $\forall\; v\in F_q^n,\;w_H(T(v))=|\{i:T(v)_i\neq0\}|=|\{i:w_i\delta(v_{\mu(i)})\neq 0\}|=|\{i: v_{\mu(i)}\neq 0\}|=w_H(v).$ Thus, $T$ preserves Hamming weight. Therefore, the Hamming weight distribution and hence the Hamming weight enumerator polynomials, $W_{C^{\perp_{Euclidean}}}(x,y)$ and $W_{C^{\perp_{W-\delta-\mu}}}(x,y)$ for the Euclidean and the $W-\delta-\mu$ duals are the same. Thus, the Mac Williams relations hold for the $W-\delta-\mu$ product as they hold for the Euclidean product. \end{proof}
	If $W$ is singleton, $\mu $ is the identity permutation and $\delta$ is the identity automorphism then obviously for all linear codes $C$ over $F_q$, the $W-\delta-\mu$ dual and the Euclidean dual coincide. But the converse is not obvious and in the following, we show that under a suitable condition on $\delta$ it holds.
	\begin{thm}\label{coincide theorem} If $\delta(xy^{-1})=xy^{-1},\;\forall\;x,y\in F_q$  then $C^{\perp_{W-\delta-\mu}}=C^{\perp_{Euclidean}}\;$ for every linear code $C$ over $F_q$if and only if $\mu=I_d,\delta=I_d $ and $ W=\{\lambda\}$ for a fixed $\lambda\in\;F_q^*$.\end{thm}
	
	\begin{proof} By Lemma \ref{4.3}, $C^{\perp_{W-\delta-\mu}}=T^{-1}(C^{\perp_{Euclidean}})$, where $T:F_q^n\longrightarrow F_q^n$ is given by $(T(y))_i=w_i\delta(y_{\mu(i)})$. Thus, $C^{\perp_{W-\delta-\mu}}=C^{\perp_{Euclidean}}$ if and only if $T^{-1}(C^{\perp_{Euclidean}})= C^{\perp_{Euclidean}}$ which is if and only if $T(C^{\perp_{Euclidean}})=C^{\perp_{Euclidean}}$, equivalently, $T$ preserves all Euclidean duals.
		
		Let $x\in\:F_q^*$. Then $(<x>^{\perp_{Euclidean}})^{\perp_{Euclidean}}=<x>$ where $<x>$ is the subspace generated by $x$. Thus, every one dimensional space is the Euclidean dual of some space, hence, preserved by $T$. Thus, $T(x)=g_xx\;\forall \; x\in F_q^n$, where $g_x\in\:F_q$ is a scalar. Now, let $x$ and $y$ be two linearly independent vectors in $F_q^n$. Then $T(x+y)=g_{(x+y)}(x+y)$ and $T(x+y)=T(x)+T(y)=g_xx+g_yy$. Equating, we get, $(g_{(x+y)}-g_x)x+(g_{(x+y)}-g_y)y=0$, which gives $g_{(x+y)}=g_x=g_y$. Thus, $g_x$ is same (say $\lambda$) for all mutually independent vectors, in particular for the basis $\{e_1,e_2,...,e_n\}.$
		
		Let $v\in F_q^n$. Then $v=\displaystyle\sum_{i=1}^{n}v_ie_i$, $v_i\in F_q$. Also, by Lemma \ref{4.3}, $T$ is $\delta$ semilinear. Thus, \[ T(v)=T(\displaystyle\sum_{i=1}^{n}v_ie_i)=\displaystyle\sum_{i=1}^{n}\delta(v_i)T(e_i)=\displaystyle\sum_{i=1}^{n}\delta(v_i)\lambda e_i=\lambda\displaystyle\sum_{i=1}^{n}\delta(v_i)e_i\tag{3}.\]
		Now, for $a\in F_q,\: T(ae_i)=\delta(a)T(e_i)=\delta(a)\lambda e_i.$ Also, $T(ae_i)=g_{ae_i}ae_i.$ Equating the two expressions we get, $\frac{\delta(a)}{a}=\frac{g_{ae_i}}{\lambda}$, thus, ${g_{ae_i}}$ is independent of $i$. Also, as $\frac{\delta(a)}{a}=\frac{\delta(b)}{b},\;\forall\;a,b\in F_q$, we have that $g_{ae_i}$ is independent of $a$ as well, and hence a constant. Thus, $\frac{\delta(a)}{a}=\mu $ for some $\mu\in F_q^*$. Putting $a=1$ gives $\mu=1$. Thus, $\delta$ is the identity automorphism. Also from (3), $T(v)=\lambda v$. Thus, $W=\{\lambda\}$ and $\mu=I_d$.
		
		Conversely, it is easy to check that if $W=\{\lambda\},\delta=I_d$ and $\mu=I_d$ then $C^{\perp_{W-\delta-\mu}}=C^{\perp_{Euclidean}}\;$ $\forall C $ .\end{proof}

	The ${W-\delta-\mu}$ dual of a linear code over $F_q$ can also be obtained explicitely from it's generator matrix. In the following proposition, we use the matrix representation from Remark \ref{3.5} for the $W-\delta-\mu$ product and use $B=D_WP$, with $D_W$ as the diagonal matrix corresponding to $W$ and $P$ the permutation matrix corresponding to $\mu$  .
	
	\begin{propo}\label{4.7} Let $C$ be a linear code over $F_q$ with a generator matrix $G$. Then $C^{\perp_{W-\delta-\mu}}=\; \Delta^{-1}(Kernel(B^TG^T))$.
	\end{propo}
	
	\begin{proof} By Remark \ref{3.5}, $<u,v>_{W-\delta-\mu}=uD_WP[\Delta(v)]^T$, where $u=[u_1\:u_2...u_n]$ and $v=[v_1\:v_2...v_n]$ are in $F_q^n$.  Also, by definition of a generator matrix, $u\in C\iff u=xG$ for some $x\in F_q^k$. Now, 
		\[
		\begin{aligned}
			v\in\; C^{\perp_{W-\delta-\mu}}\;&\iff\; <u,v>_{W-\delta-\mu}=0\; \forall u=[u_1\:u_2...u_n]\in C\\
			&\iff uD_WP[\Delta(v)]^T=0\;\forall u=[u_1\:u_2...u_n]\in C\\
			&\iff xGB[\Delta(v)]^T=0\;\forall\; x=[x_1\:x_2...x_k]\in F_q^k\\
			&\iff GB[\Delta(v)]^T=0 \\
			&\iff [\Delta(v)](GB)^T=0\\
			&\iff \Delta(v)\;\in Kernel(GB)^T\\
			&\iff v\in\; \Delta^{-1}(Kernel(B^TG^T)).
		\end{aligned}\]
		Therefore, $C^{\perp_{W-\delta-\mu}}\;=\;\Delta^{-1}(Kernel(B^TG^T))$.\end{proof}
	
	\begin{exl}\label{dual example} Consider $F_4=\{0,1,\alpha,\alpha^2\},$ where $\alpha^2=\alpha+1$, let $C$ be the linear code generated by the matrix, $G\:=\: \begin{bmatrix}
			1 & 0 & \alpha\\
			0 & 1 & \alpha^2
		\end{bmatrix}$, $W\:=\:\{1,\alpha,\alpha^2\},\; \delta(x)\;=\;x^2,\;\mu\;=\;(1\;3),\;u= (u_1,u_2,u_3),\; v=(v_1,v_2,v_3)$. Then $<u,v>_{W-\delta-\mu}\;=u_1v_3^2+\alpha u_2v_2^2+\alpha^2u_3v_1^2$. 
		Consider, \[
		\begin{aligned}
			<u,v>_{W-\delta-\mu}=0\;\forall u\in C &\iff <(1,0,\alpha),v>=0\; and\;<(0,1,\alpha^2),v>=0\\
			&\iff v_3^2+v_1^2=0\; and\; v_1^2+v_2^2=0\\
			&\iff v_1=v_2=v_3\; in\; F_4.
		\end{aligned}\]
		Thus the $W-\delta-\mu$ dual code is generated by $H\;=\: [1\;1\; 1].$
	\end{exl}
	
	\subsection{Self Orthogonality}
	\begin{dfn} A linear code $C$ over $F_q$ is $W-\delta-\mu$ self orthogonal (resp. $W-\delta-\mu$ self dual) if $C\;\subseteq C^{\perp_{W-\delta-\mu}}$ (resp. $C=C^{\perp_{W-\delta-\mu}}$)
	\end{dfn}
	We give a characterisation of $W-\delta-\mu$ self duality in terms of generator matrix of a linear code. In the following for a given matrix $G$, by $T(G)$ we mean a matrix obtained by applying $T$ on every row of $G$.
	\begin{thm}\label{dual matrix} If $G$ is a generator matrix of a linear code $C$, then $C$ is $W-\delta-\mu$ self dual if and only if $G(T(G))^T=0$.\end{thm}
	
	\begin{proof}By Lemma \ref{4.3}, $C^{\perp_{W-\delta-\mu}}=T^{-1}(C^{\perp_{Euclidean}})$, where $T:F_q^n\longrightarrow F_q^n$ is given by $(T(y))_i=w_i\delta(y_{\mu(i)})$. Thus, $C$ is $W-\delta-\mu$ self dual if and only if $T^{-1}(C^{\perp_{Euclidean}})= C$ if and only if $(C^{\perp_{Euclidean}})= T(C).$ Equivalently, generator matrix of $T(C)$ is a parity check matrix of $C$. Since $T$ is bijective the subspaces $C$ and $T(C)$ have same dimension. If $G=\begin{bmatrix}
			g^{(1)}\\
			g^{(2)}\\
			\vdots\\
			g^{(k)}
		\end{bmatrix}$, then the set $\{g^{(1)},g^{(2)},...,g^{(k)}\}$ is a basis of $C$. Consider the linear combination $\sum_{i=1}^{k}\alpha_iT(g^{(i)})=0$. Using the $\delta-$ semilinearity of $T$, we get,  $T(\sum_{i=1}^{k}\delta^{(-1)}(\alpha_i)(g^{(i)}))=0$. Thus, $\sum_{i=1}^{k}\delta^{(-1)}(\alpha_i)(g^{(i)})=0$ as $T$ is a one one mapping. This implies $\delta^{(-1)}(\alpha_i)$ is zero for each $i$ as $\{g^{(1)},g^{(2)},...,g^{(k)}\}$ is a basis of $C$. Since $\delta$ is a field automorphism, $\alpha_i=0$ for each $i$. Therefore, the set $\{T(g^{(1)}),T(g^{(2)}),...,T(g^{(k)})\}\subseteq T(C)$ is linearly independent and hence forms a basis of $T(C)$. Therefore, if $G$ generates $C$ then $T(G)$ generates the linear code $T(C)$. Now the rest of the proof follows from Lemma \ref{GH theorem}.\end{proof}
	
	Next, we give a sufficient condition for the existence of a $W$ for an arbitrary $\delta$ and $\mu$ which makes a given linear code $C$, $W-\delta-\mu$ self orthogoanl.
	
	\begin{propo}\label{Proposition:4.9} Let $C$ be a linear code of length $n$ over $F_q$ of dimension $k$. If $q>n>k^2$ then $\forall\; \mu\in\;S_n$ and $\delta\in\: Aut(F_q), $ there exists $\: W=\{w_1,w_2,\ldots,w_n\}\subseteq\;F_q^* $ such that $C$ is $W-\delta-\mu$ self orthogonal.\end{propo}
	
	\begin{proof} Let \[G=
		\begin{bmatrix}
			g^{(1)}\\
			g^{(2)}\\
			\vdots\\
			g^{(k)}
		\end{bmatrix}_{k\times n}
		\] be a generator matrix for $C$. For any vector $v=[v_1,v_2,...,v_n]\in F_q^n$, consider the following set of $k^2$ homogenous equations in $v_1,v_2,...,v_n$,
		\[\displaystyle\sum_{i=1}^{n}v_ig_i^{(r)}\delta(g_{\mu(i)}^{(s)})=0,\; 1\leq r,s\leq k. \tag{4}\]
		Equivalently, $Av^T=0$ where $A$ is a $k^2\times n$ matrix given by $A$ $=[g_i^{(r)}\delta(g^{(s)}_{\mu(i)})]_{1\leq r,s\leq k,1\leq i\leq n}$.
		The solution space of the system of equations (4) is given by the null space of the matrix $A$. If we denote the nullity of $A$ by $\nu$ and the rank of $A$ by $\rho$, then $n>k^2$ implies that $\nu=n-\rho\geq n-k^2\geq 1.$ Then there exists a $v^T\neq\:0$ such that $Av^T=0.$ Since, nullity of $A$ is $\nu$, there are $q^\nu$ vectors in $Kernel(A)$. So, the number of vectors in $Kernel(A)$ with 0 in the $i^{th}$ place is $q^{(\nu-1)}$. Thus, number of vectors with at least one zero coordinate in $Kernel(A)$ is $nq^{(\nu-1)}< q^\nu $ because $n<q$. Thus, there exists a $v^T= [v_1,v_2,\ldots,v_n]^T$ with all $v_i$'s non zero satisfying the system of equations (4).
		Therefore, there exists $W=\{w_1,w_2,...,w_n\}\subseteq F_q^*$, where for each $i\;w_i=v_i$ such that\[
			\displaystyle\sum_{i=1}^{n}w_ig_i^{(r)}\delta(g_{\mu(i)}^{(s)})=0,\; 1\leq r,s\leq k\]
			This implies that$	<g^{(r)},g^{(s)}>_{W-\delta-\mu}\;=0\;\forall\; 1\leq r,s\leq k.$ Hence, $C$ is $W-\delta-\mu$ self orthogonal.

		\end{proof}
	
	\begin{coro} If the $W-\delta-\mu$ product is symmetric for the given choice of $\delta$ and $\mu$, then the system of equations in $w_1,w_2,\ldots,w_n$ has at the most $\frac{k(k+1)}{2}$ linearly independent equations. Moreover, if we take $q>n>\frac{k(k+1)}{2}$ then there exists a set $W$ such that $C$ is $W-\delta-\mu$ self orthogonal.\end{coro}
	
	\begin{exl} Consider, $k=2, n=5, \delta= I_d$ in $F_7$ and $ \mu$ as identity permutaion. Then, $\frac{k(k+1)}{2}=3<5<7.$ Thus, by above proposition, for every linear code $C$ over $F_7$ of length 5 and dimension 2, there exists $\;\{w_1,w_2,w_3,w_4,w_5\}\subseteq F_7^*$ such that $C$ is $W-\delta-\mu$ self orthogonal.
		For instance, if we take $C$ to be the linear code generated by the matrix $G=\begin{bmatrix}
			1 & 0 & 1 & 1 & 2\\
			0 & 1 & 2 & 3 & 4
		\end{bmatrix}$ then we get the following system of linear equations from the $W-\delta-\mu$ self orthogonality condition:
		\[\begin{aligned}
			w_1+w_3+w_4+4w_5=0\\
			w_2+4w_3+2w_4+2w_5=0\\
			2w_3+3w_4+w_5=0
		\end{aligned}
		\]
		Clearly, $W=\{1,1,1,2,6\}$ is one such solution.\end{exl}
		\begin{exl}\label{Better example} Let $G=
			\begin{bmatrix}
				1&0&a\\
				0&1&a+1
			\end{bmatrix}
			$ be a generator matrix over the finite field $F_4$, with $n=3$, $k=2$, $q=4$ 
			We observe that for this generator matrix, the condition of the Proposititon \ref{Proposition:4.9}, $q>n>k^2$ is not satisfied. By SageMath computations \cite{Sage} on $G$, we obtain the following informations with respect to $W-\delta-\mu$ products. 
			
							\begin{center}
					\renewcommand{\arraystretch}{1.3}
					\begin{tabular}{|c|c|}
						\hline
						Total $W-\delta-\mu$ products & $324$\\
						\hline
						$W-\delta-\mu$ LCD Codes & $252$\\
						\hline
						$W-\delta-\mu$ Self Orthogonal Codes & $0$\\
						\hline
						Total $\delta$ products & $2$\\
						\hline
						$\delta$ LCD codes & $1$\\
						\hline
						$\delta$ self orthogonal codes & $0$\\
						\hline
					\end{tabular}
				\end{center}
				
				\[
				\begin{aligned}
					\text{Percentage of LCD Codes}
					&=77.78\%,\\
					\text{Percentage of Self Orthogonal Codes}
					&=0\%.					
				\end{aligned}
				\]
			\end{exl}
				\vspace{2ex}
			 Thus, there are no $W-\delta-\mu$ self orthogonal codes in this case.
			 \begin{exl}
			 Consider the following generator matrix over $F_4$.
				\[
				G'=
				\begin{bmatrix}
					1&0&a\\
					a&0&a+1
				\end{bmatrix}
				\]
				Then, we observe that $n=3$, $q=4$ and $k=1$ which satisfy the conditions of the Proposition \ref{Proposition:4.9}, that is, $q>n>k^2$. By Sagemath computations, we get the following informations with respect to $W-\delta-\mu$ products.
				\begin{center}
					\renewcommand{\arraystretch}{1.3}
					\begin{tabular}{|c|c|}
						\hline
						Total $W-\delta-\mu$ products & $324$\\
						\hline
						$W-\delta-\mu$ LCD Codes & $288$\\
						\hline
						$W-\delta-\mu$ Self Orthogonal Codes & $36$\\
						\hline
						Total $\delta$ products & $2$\\
						\hline
						$\delta$ LCD codes & $1$\\
						\hline
						$\delta$ self orthogonal codes & $1$\\
						\hline
					\end{tabular}
				\end{center}
				
				\[
				\begin{aligned}
					\text{Percentage of $W-\delta-\mu$ LCD Codes}
					&=88.89\%,\\
					\text{Percentage of $W-\delta-\mu$ Self Orthogonal Codes}
					&=11.11\%.
				\end{aligned}
				\]
				We find that there are 36 $W-\delta-\mu$ self orthogonal codes in this case, validating Proposition \ref{Proposition:4.9}.
			\end{exl}

	Note that in case of $\delta$- product over $F_4$ we have only two cases, that is, $\delta$ is either identity map or frobenius map. However, for a code of length $3$ over $F_4$ there are 324 distinct choices for $W-\delta-\mu$ products with many of them yielding $W-\delta-\mu$ LCD and $W-\delta-\mu$ self orthogonal codes.
	
	The condition for self duality can also be given in terms of the matrix product representation associated with the $W-\delta-\mu$ product.  If $G=[g_{ij}]$ then we denote the matrix $[\delta(g_{ij})]$ by $\Delta(G)$. We observe that if $x\in F_q^k$ then $\Delta(xG)=\Delta(x)\Delta(G)$.
	
	\begin{propo} Let $C$ be a linear code over $F_q$ with a generator matrix $G$. Then $C$ is $W-\delta-\mu$ self orthogonal if and only if $GB\Delta(G)^T=0$. Additionally, if $Rank(G)=Nullity(B^TG^T)$ then $C$ is $W-\delta-\mu$ self dual.
	\end{propo}
	
	\begin{proof} Consider, $C$ is $W \delta\mu$ self orthogonal i.e. \[
		\begin{aligned}
		    C\subseteq C^{\perp_{W-\delta-\mu}}
			&\iff C\subseteq \Delta^{-1}(Kernel(B^TG^T)),\;by\; Proposition \;\ref{4.7}\\
			&\iff \Delta(C)\subseteq (Kernel(B^TG^T))\\
			&\iff \Delta(u)B^TG^T=0\:\forall\:u\in C\\
			&\iff \Delta(xG)(B^TG^T)=0\:\forall x\in F_q^k\\
			&\iff \Delta(x)\Delta(G)B^TG^T=0\:\forall\:x\in F_q^k\\
			&\iff x\Delta(G)B^TG^T=0\:\forall\:x\in F_q^k\\
			&\iff \Delta(G)B^TG^T=0\\
			&\iff GB[\Delta(G)]^T=0.
		\end{aligned}\]Further, let $Rank(G)=Nullity(B^TG^T).$ Since,  $\Delta$ is bijective, we get
		$Rank(G)=Dim(C)$  $=Nullity(\Delta^{-1}(B^TG^T))=Dim(\Delta^{-1}(Kernel(B^TG^T))).$ Thus, $C\subseteq \Delta^{-1}(Kernel(B^TG^T))$ implies that $ C=\Delta^{-1}(Kernel(B^TG^T))=C^{\perp_{W-\delta-\mu}}.$ Hence, $C$ is $W-\delta-\mu$ self dual.
	\end{proof}
	
	Now, we give a relation between self orthogonalities and self dualities in case of usual product and $W-\delta-\mu$ product. For a linear code $C$ over $F_q$, $W=\{w_1,w_2,...,w_n\}\subseteq F_q^*$, $\delta\in Aut(F_q)$, $\mu\in S_n$ and $u= (u_i)_n\in C$ we denote, $Wu=(w_1u_1,w_2u_2,\ldots,w_nu_n)$, $WC=\{Wu:u\in C\}$,   $\mu(u)=(u_{\mu(1)},u_{\mu(2)},\ldots,u_{\mu(n)})$, $\mu(C)=\{\mu(u):u\in C\}$, $\Delta^{-1}(u)= (\delta^{-1}(u_1),\delta^{-1}(u_2),\ldots,\delta^{-1}(u_n))$ and $\Delta^{-1}(C)=\{\Delta^{-1}(u):u\in\:C\}.$ 
	\begin{propo}\label{SO and SD} A linear code $C$ over $F_q$ is $W-\delta-\mu$ self orthogonal if and only if $\mu(C)\subseteq\:\Delta^{-1}(WC)^{\perp_{Euclidean}}$. Further, $C$ is $W-\delta-\mu$ self dual if and only if $\mu(C)=\:\Delta^{-1}(WC)^{\perp_{Euclidean}}$.  \end{propo}
	
	\begin{proof} We have,
		\[
		\begin{aligned}
			C\subseteq C^{\perp_{W-\delta-\mu}} &\iff\;<u,u>_{W-\delta-\mu}=0\;\forall\;u\in\:C\\
			&\iff \displaystyle\sum_{i=1}^{n}w_iu_i\delta(u_{\mu(i)})=0\\
			&\iff\displaystyle\sum_{i=1}^{n}\delta(\delta^{-1}(w_iu_i)u_{\mu(i)})=0\\
			&\iff \delta(\displaystyle\sum_{i=1}^{n}\delta^{-1}(w_iu_i)u_{\mu(i)})=0\\
			&\iff \displaystyle\sum_{i=1}^{n}\delta^{-1}(w_iu_i)u_{\mu(i)}=0\\
			&\iff <\Delta^{-1}(Wu),\mu(u)>_{Euclidean}=0\\
			&\iff \Delta^{-1}(WC)\perp_{Euclidean} \mu(C)\\
			&\iff \mu(C)\subseteq\Delta^{-1}(WC)^{\perp_{Euclidean}}.
		\end{aligned}
		\]
		Analogously, it follows that $C= C^{\perp_{W-\delta-\mu}}\iff \mu(C)=\Delta^{-1}(WC)^{\perp_{Euclidean}}$.\end{proof}
	
	\section{$W-\delta-\mu$ Dual containing and LCD codes}
	\begin{dfn}\label{dual containing} A linear code $C$ over $F_q$ is $W-\delta-\mu$ dual containing if $C^{\perp_{W-\delta-\mu}}\subseteq C$.\end{dfn}
	
	\begin{propo} If $C$ is $W-\delta-\mu$ dual containing code of length $n$ over $F_q$, then $|C|\ge q^{n/2}.$\end{propo}
	
	\begin{proof} By Proposition \ref{4.4}, $dim(C)+dim(C^{\perp_{W-\delta-\mu}})=n$. The proof now follows directly from Definition \ref{dual containing}.
	\end{proof}
	
	\begin{propo}\label{5.3} A linear code $C$ over $F_q$ is a $W-\delta-\mu$ dual containing if and only if $C^{\perp_{Euclidean}}\subseteq T(C)$.\end{propo}
	
	\begin{proof}By definition \ref{dual containing}, $C$ is dual containing if and only if $C^{\perp_{W-\delta-\mu}}\subseteq C$. This is if and only if $C^{\perp_{W-\delta-\mu}}=T^{-1}(C^{\perp_{Euclidean}})$, by Lemma \ref{4.3}, where $T:F_q^n\longrightarrow F_q^n$, given by $(T(y))_i=w_i\delta(y_{\mu(i)})$ is a $\delta$-semilinear bijection. Equivalently, $C^{\perp_{Euclidean}}\subseteq T(C)$.\end{proof}
		
	Below we give examples which show that a linear code can be $W-\delta-\mu$ dual containing while not Euclidean dual containing and vice versa.
	
	\begin{exl}
		
		Consider the binary field $F_2$ and a linear code $C$ of length $3$ over it with first coordinate zero for each codeword. Let $W=\{1,1,1\},\mu(i)=i+1,\delta=I_d$. Then $C^{\perp_{Euclidean}}=\{(0,\;0,\;0),(1,\;0,\;0)\}\nsubseteq C$ but $C^{\perp_{W-\delta-\mu}}=\{(0,\;0,\;0),(0,\;1,\;0)\}\subseteq C$.\end{exl}
	
	\begin{exl} Consider the length $4$ linear code $C=\{(a,a,b,b)\;:\;a,b\in F_2\}$ over the binary field $F_2$, then $C$ is Euclidean self dual and hence Euclidean dual containing. Let $W=\{1,1,1\},\delta=I_d,\mu(i)=i+1$. Then $C^{\perp_{W-\delta-\mu}}=\{(a,b,b,a)\;:\;a,b\in F_2\}\nsubseteq C$. For instance, $(1,0,0,1)\in C^{\perp_{W-\delta-\mu}}$ but $(1,0,0,1)\notin C$.
	\end{exl}
	
	As in the case of self orthogonal codes, we can talk of existence result for $W$ for a linear code $C$ to be $W-\delta-\mu$ dual containing.
	
	\begin{propo}\label{5.8}
		
		There exists a $W=\{w_1,w_2,...,w_n\}\subseteq F_q^*$ such that $C^{\perp_{W-\delta-\mu}}\subseteq C $ if and only if there exist codewords $g^{(1)},g^{(2)},...,g^{(s)}$ such that  for all $i=1,2,...,n,$ ${\delta(g_i^{(k)})}\neq 0$, where $s$ is the dimension of the Euclidean dual of $C$.
	\end{propo} 
	
	\begin{proof} Let $B=\{v^{(1)},v^{(2)},...,v^{(s)}\} $ be a basis of $C^{\perp_{Euclidean}}$. Suppose there exist codewords $\;g^{(1)},g^{(2)},...,g^{(s)}$ such that $\forall\;i=1,2,...,n$, ${\delta(g_i^{(k)})}\neq 0$, where $s$ is the dimension of the Euclidean dual of $C$. This implies ${\delta(g_{\mu^(i)}^{(k)})}\neq 0$. Then for any $i^{th}$ component of $v^{(k)}$ i.e. $v_i^{(k)}$, we take $w_i = \frac{v_i^{(k)}}{\delta(g_{\mu(i)}^{(k)})}$. Thus, for every $i$  $v_i^{(k)}=w_i\delta(g^{(k)}_{\mu(i)})=T(g^{(k)})_i.$ It follows that $v^{(k)}=T(g^{(k)})\in T(C)$. Hence, by linearity, $C^{\perp_{Euclidean}}\subseteq T(C)$. Thus, by Proposition \ref{5.3} $C^{\perp_{W-\delta-\mu}}\subseteq C$.
		
		Conversely, let $W=\{w_1,w_2,...,w_n\}\subseteq F_q^* $ such that $C^{\perp_{W-\delta-\mu}}\subseteq C$. Then by Proposition \ref{5.3}, $C^{\perp_{Euclidean}}\subseteq T(C)$. Since $B$ is a basis, $v^{(i)}$ is non zero. Also, since $B\subseteq C^{\perp_{Euclidean}}\subseteq T(C)\;, $ there exist $\; g^{(1)},g^{(2)},...,g^{(s)}$ such that $v^{(k)}_i =T(g^{(k)})_i=w_i\delta(g^{(k)}_{\mu(i)})\neq 0\;\forall\;i$. Hence, $\delta(g^{(k)}_i)\neq 0$.
	\end{proof}
	
	Another classical notion of relation between a linear code and its dual is that of linear complementarity, i.e., they intersect trivially having only zero in common. Such a code is called a LCD code, that is, linear code with complementary dual \cite{Massey}. In \cite{Sendrier}, it was shown that they meet the Gilbert Varshamov bound. In \cite{Yang} and \cite{Esmaeili} certain conditions were found for cyclic and some quasicyclic codes respectively to be LCD codes.
	
	\begin{dfn}\label{LCD} A linear code $C$ over $F_q$ is $W-\delta-\mu$ LCD code if $C\cap C^{\perp_{W-\delta-\mu}}=\{0\}$.
	\end{dfn}
	
	A linear code may not be LCD with respect to the Euclidean product, but with appropriate choices for $W,\;\delta$ and $\mu$, it can be made $W-\delta-\mu$ LCD. This is illustrated in the following example:
	
	\begin{exl}
		Consider the linear code $C=\{(0,0,0),(1,0,1)\}\subseteq\:F_2^3$. Then $<(1,0,1),(1,0,1)>_{Euclidean}=0$ and hence $(1,0,1)\in\:C^{\perp_{Euclidean}}$. Thus $C$ is not a LCD code with respect to the usual product. Now, consider $W=\{1,1,1\},\:\delta=I_d$ and $\mu=(1\;2\;3)\in\:S_3$. We have $<(1,0,1),(1,0,1)>_{W-\delta-\mu}=1\neq0$. Hence, $C$ is $W-\delta-\mu$ LCD. 
	\end{exl}
	
	Recall the transformation $T:F_q^n\longrightarrow F_q^n$ from Lemma \ref{4.3} with $(T(y))_i=w_i\delta(y_{\mu(i)})$ such that $C^{\perp_{W-\delta-\mu}}=T^{-1}(C^{\perp_{Euclidean}})$. So, we get the equivalent condition for $W-\delta-\mu$ LCD codes, that is, $C$ is a $W-\delta-\mu$ LCD code if and only if its Euclidean dual intersects trivially with its $T-$image. 
	
	The LCD codes have been characterised in terms of their generator matrices as follows:
	
	\begin{propo}\cite[Proposition 1]{Massey} If $G$ is a generator matrix for the (n, k) linear code $C$, then $C$ is an LCD code if and only if the $k\times k$ matrix $GG^T$ is nonsingular.
	\end{propo}
	
	We have a similar characterization of $W-\delta-\mu$ LCD codes. In the following, $\Delta^{-1}(a)=(\delta^{-1}(a_i))_k\:\forall a=(a_i)_k\in F_q^k.$
	\begin{propo}\label{LCD matrix} Let $G$ be a generator matrix of $C$. Then $C$ is $W-\delta-\mu$ LCD if and only if $M=G(T(G))^T$ is a non singular matrix.
	\end{propo}
	\begin{proof} Let $C$ be not $W-\delta-\mu$ LCD. Then by Definition \ref{LCD}, there exists $\; x\neq 0$ in $C$ such that $x\in C\cap C^{\perp_{W-\delta-\mu}}.$ Following the proof of Lemma \ref{4.3}, this implies $<y,T(x)>_{Euclidean}=0$ for all $y\in C$. Thus, $T(x)(G)^T=0.$ Further, $x=aG$ for some $a\neq 0$ as $G$ is a generator matrix for $C$. Hence, $T(aG)G^T=0$. It is easy to verify that $T(aG)=\Delta^{-1}(a)T(G)$. Thus, $\Delta^{-1}(a)T(G)G^T=0$, with $\Delta^{-1}(a)\neq0$, as $a\neq 0$ and $\Delta$ is bijective. Hence, $T(G)G^T$ is singular. Therefore, $M=G(T(G))^T$ is singular. 
		
		Conversely, suppose that $M$ is singular then there exists $a\neq 0$ in $F_q^k$ such that  $M(\Delta^{-1}(a))^T=0.$ This implies $G(T(G))^T(\Delta^{-1}(a))^T=0$. Taking transpose, $(\Delta^{-1}(a))T(G)G^T=0.$ Thus, $T(aG)G^T=0$. Since $G$ is a generator matrix which is always injective, therefore, $aG=x$ is a non zero codeword of $C$. Hence, $T(x)G^T=0$, where $x\neq0$ in $C$. This implies, $<y,T(x)>_{Euclidean}=0$ for each row $y$ of $G$. Hence, by linearity, $<y,T(x)>_{Euclidean}=0$ for each $y\in C$ as $G$ generates $C$. Following the proof of Lemma \ref{4.3}, $x\in C^{\perp_{W-\delta-\mu}}$. Hence, $x\in C\cap C^{\perp_{W-\delta-\mu}}$. Therefore, $C$ is not $W-\delta-\mu$ LCD.
	\end{proof}
	
	We now give a sufficient condition for the existence of a set $W$ such that a given linear code $C$ is $W-\delta-\mu$ LCD for any choice of $\delta$ and $\mu$. First, we state $Schwartz\; Zippel\; Lemma$ in a special form.
	
	\begin{lem}\cite[Schwartz Zippel Lemma]{Atserias}\label{5.13}
		Let $P(x_1,x_2,...,x_n)$ be a non zero multivariate polynomial in $n$ variables over an integral domain $R$, of degree $d$. Then $|Z(P)\cap S^n|\leq d|S|^{n-1}$, where $Z(P)$ is the zero set of $P$ and $S$ is any non empty subset of $R$.
	\end{lem}
	
	\begin{thm}\label{5.12} Let $C$ be a linear code of length $n$ over $F_q$ of dimension $k$. If $q>k+1$, then there exists $\; W=\{w_1,w_2,...,w_n\}\subseteq F_q^*$ such that $C$ is $W-\delta-\mu$ LCD for any choice of $\delta$ and $\mu$.
	\end{thm}
	\begin{proof} Let $G=\begin{bmatrix}
			g_1^1 & g_2^1 & \cdots & g_n^1\\
			g_1^2 & g_2^2 & \cdots & g_n^2\\
			\vdots & \vdots\ & \ddots & \vdots\\
			g_1^k & g_2^k & \cdots & g_n^k
		\end{bmatrix}$ be a generator matrix of $C$. Consider an arbitrary set $W=\{w_1,w_2,...,w_n\}\subseteq F_q^*$. Then for any $\delta$, $\mu$, with $T(y)=w_i\delta(y_{\mu(i)})$, we have, $(T(G))^T=\begin{bmatrix}
			w_1\delta(g^1_{\mu(1)}) & w_1\delta(g^2_{\mu(1)}) & \cdots & w_1\delta(g^k_{\mu(1)})\\
			w_2\delta(g^1_{\mu(2)}) & w_2\delta(g^2_{\mu(2)}) & \cdots & w_2\delta(g^k_{\mu(2)})\\
			\vdots & \vdots & \ddots & \vdots\\
			w_n\delta(g^1_{\mu(n)}) & w_n\delta(g^2_{\mu(n)}) & \cdots & w_n\delta(g^k_{\mu(n)})\\
		\end{bmatrix} $. Thus, $(G(T(G))^T)_{ij}= \displaystyle\sum_{k=1}^{n} w_kg_k^{(i)}\delta (g_{\mu(k)}^{(j)})$ which is a linear expression in  $w_1,w_2,...,w_n$ over $F_q$. Thus, determinant of $G(T(G))^T$ is a polynomial in the variables $w_i$, say, $P(w_1,w_2,...,w_n)$ of degree at most $k$, the size of $G(T(G))^T$. Using Proposition \ref{LCD matrix}, $C$ is $W-\delta-\mu$ LCD if and only if $G(T(G))^T$ is non singular if and only if $Det(G(T(G))^T)\neq 0$. Equivalently, $P(w_1,w_2,...,w_n)\neq 0$. Since, $P(w_1,w_2,...,w_n)$ is a multivariate polynomial in $n$ variables over $F_q$ of degree $k$, using Lemma \ref{5.13} and $S=F_q^*$ with $|S|=q-1$, we get that $P$ has at most $k(q-1)^{n-1}$ zeros in $F_q^*$. Since, $q>k+1$, that is, $k<q-1$, therefore, $k(q-1)^{n-1}<(q-1)^n$. Thus, there exists a set $\;\{w_1,w_2,...,w_n\}\subseteq F_q^*$ such that $P(w_1,w_2,...,w_n)\neq 0$. Thus, $C$ is $W-\delta-\mu$ LCD.
	\end{proof}
	Note that the above result shows that for a large enough field, having size more than the dimension of a linear code $C$, we can always have a $W$ which makes $C$ a $W-\delta-\mu$ LCD code, for arbitrary $\delta$ and $\mu$.
		\section{$W$-$\delta$-$\mu$ dual of some classes of codes} 	Now, we explore the duals of some particular classes of codes with respect to the $W-\delta-\mu$ product. The symbols used in this section have the same meaning as in Proposition \ref{SO and SD}.
	\begin{rmk}
		\begin{enumerate}
			
			\item {\textbf{Repetition codes}}: Let $C$ be a repetition code of length $n$ over $F_q$, that is, $C=\{(x,x,\ldots,x)\::\:x\in\:F_q\}$. Then
			\[
			\begin{aligned}
				C^{\perp_{W-\delta-\mu}}&=\{y\in\:F_q^n\::\:<c,y>_{W-\delta-\mu}=0\;\forall\;c\in\:C\}\\
				&=\{(y_1,y_2,\ldots,y_n)\in\:F_q^n\;:\;\displaystyle\sum_{i=1}^{n}w_ix\delta(y_{\mu(i)})=0\;\forall\;x\in F_q\}\\
				&=\{(y_1,y_2,...,y_n)\in\:F_q^n\;:\;\displaystyle\sum_{i=1}^{n}w_i\delta(y_{\mu(i)})=0\}\\
				&=\{(y_1,y_2,...,y_n)\;:\; \displaystyle\sum_{i=1}^{n}\delta^{-1}(w_{\mu^{-1}(i)})y_i=0\}\\
				&=\{y\in\:F_q^n\;:\;<y,\Delta^{-1}(\mu^{-1}(W))>_{Euclidean}=0\}\\
				&=(\Delta^{-1}(\mu^{-1}(W)))^{\perp_{Euclidean}
				}
			\end{aligned}
			\]
			
			Thus, we see that the $W-\delta-\mu$ dual of a repetition code is simply the Euclidean orthogonal space of the successively taken inverse images under $\mu$ and $\Delta$ of the vector $(w_1,w_2,...,w_n)\subseteq (F_q^*)^n$
			
			\item{\textbf{Binary codes}}: If the code alphabet is the binary field $F_2$ then the only possibilities for $W$ and $\delta$ are $\{1,1\}\subseteq F_2^*$ and $I_d$ on $F_q$.Thus, the $W-\delta-\mu$ product differs from the Euclidean product only by a permutation on the second coordinate.
		\end{enumerate}
	\end{rmk}
	Following \cite[Proposition 2.3]{Dinh}, the Euclidean dual of a $\lambda$-constacylic code is a $\lambda^{-1}-$constacyclic code. Further, in \cite[Proposition 3.2]{Dinh}, authors show that the $\delta$-dual code of a $\lambda-$constacyclic code of length $n$ over a finite field $F$ is a $\delta^{-1}(\lambda^{-1})$- constacyclic code of length $n$ over $F$. In case of $W-\delta-\mu$ product the $W-\delta-\mu$ dual of a constacyclic code need not be constacyclic as illustrated in Example \ref{consta example}. However, in the following proposition we show that if $\mu$ is the identity permutation then $W-\delta-\mu$ dual of a $\lambda$- constacyclic code is invariant under a map which involves scaling along with constacylic shift.
			\begin{propo} Let $C$ be a $\lambda$ constacyclic code over $F_q$ for a given $\lambda\in\:F_q^*$ and $\mu$ be the identity permutation on $n$ symbols. Then $C^{\perp_{W-\delta-\mu}}$ is invariant under the operator $\zeta:F_q^n\longrightarrow F_q^n$ given by $\zeta(y)=\tau_{\delta^{-1}(\lambda^{-1})}(W'y)$, where $W'=\Delta^{-1}(w_1w_2^{-1},...,w_nw_1^{-1})$.
			\end{propo}
			\begin{proof}
	    	 Since $C$ is a $\lambda-$constacylic code, therefore for all $x=(x_1,x_2,...,x_n)\in\:C\:, \tau_\lambda(x)=(\lambda x_n,x_1,...,x_{n-1})\in\:C$ which implies $\tau_\lambda^{n-1}(x)=(\lambda x_2,\lambda x_3,...,\lambda x_n, x_1)\in\:C$. Thus, $y\in C^{\perp_{W-\delta-\mu}}$ if and only if $y\in (\tau_\lambda^{n-1}(C))^{\perp_{W-\delta-\mu}}.$ We have,
			
			$$y\in\:C^{\perp_{W-\delta-\mu}} \iff y\in (\tau_\lambda^{n-1}(C))^{\perp_{W-\delta-\mu}}\\$$	$$\iff<\tau_\lambda^{n-1}(x),y>_{W-\delta-\mu}=0 \:\forall\: x\in\:C\\$$
			$$\iff\:w_nx_1\delta(y_{\mu(n)})+\displaystyle\sum_{i=2}^{n}w_{i-1}\lambda x_i\delta(y_{\mu(i-1)})=0\\$$
			$$\iff\lambda^{-1}w_nx_1\delta(y_{n})+\displaystyle\sum_{i=2}^{n}w_{i-1}x_i\delta(y_{i-1})=0\\$$
			$$\iff x_1\delta(\delta^{-1}(\lambda^{-1}w_n)y_{n})+x_2\delta(\delta^{-1}(w_1)y_{1})+...\\$$
			$$\;\;\;\;\;\;\;\;\;\;\;\;+x_n\delta(\delta^{-1}(w_{n-1})y_{n-1})=0\\$$
			$$\iff w_1x_1\delta(\delta^{-1}(\lambda^{-1})\delta^{-1}(w_nw_1^{-1}))y_{n})+w_2x_2\delta(\delta^{-1}(w_1w_2^{-1})y_{1})+...\\$$
			$$\;\;\;\;\;\;\;\;\;\;\;\;+w_nx_n\delta(\delta^{-1}(w_{n-1}w_n^{-1})y_{n-1})=0\\$$
			$$\iff<x,\tau_{\delta^{-1}(\lambda^{-1})}(W'y)>_{W-\delta-\mu}=0\\$$
			$$\iff <x,\zeta(y)>_{W-\delta-\mu}=0,$$ where $W'=\Delta^{-1}(w_1w_2^{-1},...,w_nw_1^{-1})$. Therefore, $C^{\perp_{W-\delta-\mu}}=\zeta(C^{\perp_{W-\delta-\mu}}).$	
		\end{proof}
	We observe here that the $W-\delta-\mu$ dual of a $\lambda$ constacyclic code need not be a constacyclic code. Below, we give an example to illustrate this point.
	
	\begin{exl}\label{consta example} Consider $F_5=\{0,1,2,3,4\}$ with mod 5 arithmetic, $W=\{1,2,1,1\}\subseteq F_5^*, \mu=(1 \;2\;3\;4)\in S_4$ and $\delta=I_d \in Aut (F_5)$. For the linear repetition code (which is also cyclic and hence constacyclic) $C=\{(a,a,a,a)\;:\;a\in F_5\}$, $y\in F_5^4$ is in $C^{\perp_{W-\delta-\mu}}$ if and only if $\forall\;x=(a,a,a,a)\in C\;
		<x,y>_{W-\delta-\mu}=0$ which is if and only if $$y_1+y_2+2y_3+y_4 = 0\;(mod\;5).$$ Clearly, $y=(4,1,0,0)$ is one such solution and hence a member of the $W-\delta-\mu$ dual, but any $\lambda$ constacyclic shift of $y$, i.e., $\tau_\lambda(y)$ is equal to $(0,4,1,0), \forall \;\lambda\in F_5$ which is not in the  dual since $(0,4,1,0)$ does not satisfy the required condition. Therefore, $C^{\perp_{W-\delta-\mu}}$ is not $\lambda$- constacyclic.
	\end{exl}
	
	\section{$W-\delta-\mu$ Dual of codes over semi-simple rings}
	We now generalise in a natural way the definition of $W-\delta-\mu$ product to the case when the alphabet is a semisimple ring and the ambient space is $R^n$ as $R$ module.  
	\begin{rmk}\label{decomposition}
		Let $R=\displaystyle\bigoplus_{j=1}^k F_{q_j}$ be a finite commutative semisimple ring. Then this is a standard result \cite[Proposition 3.1]{Dinh2} that if $C$ is a linear code over $R$ of length $n$, then $C=\displaystyle\bigoplus_{j=1}^{k}C_j$, where $C_j$ is a linear code of length $n$ over $F_{q_j}$, for $j=1,2,...,k$.
	\end{rmk}
	
	\begin{dfn}
	 Let $W=\{W^{(1)},W^{(2)},...,W^{(k)}\},\; \mu=\{\mu^{(1)},\mu^{(2)},...,\mu^{(k)}\},\; \delta=(\delta^{(1)},\delta^{(2)},...,\delta^{(k)}),$ where $W^{(i)}=\{w^{(i)}_1,w^{(i)}_2,...,w^{(i)}_n\}\subseteq F_{q_i}^*,\;\mu^{(i)}\in S_n,\delta^{(i)}\in\:Aut(F_{q_i}),\: x=(x^{(1)},x^{(2)},...,x^{(k)}), y=(y^{(1)},y^{(2)},...,y^{(k)})\in R^n$, where $x^{(i)},y^{(i)}\in F_{q_i}^n$. We define, $<x,y>_{W-\delta-\mu}:R^n\times R^n\longrightarrow R$ by $$<x,y>_{W-\delta-\mu}=(<x^{(1)},y^{(1)}>_{W^{(1)}-\delta^{(1)}-\mu^{(1)}},.\:.\:.\:,<x^{(k)},y^{(k)}>_{W^{(k)}-\delta^{(k)}-\mu^{(k)}}).$$
	\end{dfn}
	The following results are valid for $W-\delta-\mu$ product anlogously as they are valid for the $\delta$-product in \cite{Dinh}. So, their proofs are also analogous to that.
	In the following result, we generalize \cite[Proposition 3.4]{Dinh} and its proof is analogous.
	
	\begin{propo}\label{8.2} Let $C=\displaystyle\bigoplus_{j=1}^{k}C_j$ be a code of length $n$ over $R$. Then $C^{\perp_{W-\delta-\mu}}=\displaystyle\bigoplus_{j=1}^{k}C_j^{\perp_{W^{(j)}-\delta^{(j)}-\mu^{(j)}}}$.\end{propo}
	A result \cite[Theorem 3.5]{Dinh} is a consequence of the following result and its proof is on the same line to that.
	
	\begin{thm}\label{followings} Let $C$ be a linear code over $R$ of length $n$. Then followings hold true:
		\begin{enumerate}[(1)]
			\item $C$ is $W-\delta-\mu$ self orthogonal (self dual) if and only if each constituent code $C_j,\;j=1,2..,k$ is $W^{(j)}-\delta^{(j)}-\mu^{(j)}$ is self orthogonal (self dual).
			\item $C$ is $W-\delta-\mu$ dual containing if and only if each constituent code $C_j$ is $W^{(j)}-\delta^{(j)}-\mu^{(j)}$ dual containing.
			\item $C$ is a $W-\delta-\mu$ LCD code if and only if each constituent code $C_j$ is a $W^{(j)}-\delta^{(j)}-\mu^{(j)}$ LCD code.
		\end{enumerate}
			\end{thm}
			\begin{coro}\cite[Theorem 3.5]{Dinh}
				Let $\lambda$ be a unit of $R$, and $C=\displaystyle\bigoplus_{i=0}^{k-1}C_i$
				a $\lambda-constacyclic$ code of length $n$ over $R$. The followings hold true:
				\begin{enumerate}[(1)]
					
					\item $C$ is $\delta_R$-dual-containing if and only if all of the codes $Ci$ are $\delta_i$
					-dual-containing for all $i=0,...,k-1.$
					\item  $C$ is $\delta_R$-self-orthogonal if and only if all of the codes $C_i$ are $\delta_i$-self-orthogonal for all $i=0, ...,k-1.$
					\item $C$ is $\delta_R$-self-dual over $R$ if and only if all the codes $C_i$ are $\delta_i$-self-dual over $F_i$	for all $i=0,...,k-1.$
					\item  $C$ is $\delta_R$-LCD over $R$ if and only if all the codes $C_i$	are $\delta_i$-LCD over $F_i$ for all $i=0, ...,k-1.$
				\end{enumerate}
			\end{coro}

	\begin{coro}
		
		If $q^{(j)}>n>k^{{(j)}^2}\;\forall\;j=1,2,..,k$, where $k^{(j)}$ is the dimension of the code $C_j$ over $F_{q_j}$ then there exists $W=\{W^{(1)},W^{(2)},...,W^{(k)}\}$ such that $C$ is $W-\delta-\mu$ self orthogonal $\forall\; \delta\in\;\displaystyle\bigoplus_{j=1}^{k} Aut(F_{q_j})$ and $\mu=\{\mu^{(1)},\mu^{(2)},...,\mu^{(k)}\}$ with $\mu^{(j)}\in \; S_n$.
	\end{coro}
	
	\begin{proof} It is a consequence of Theorem \ref{followings}(1) and follows from the decomposition of the $W-\delta-\mu$ dual space over $R$ given in Proposition \ref{8.2}. The existence part follows from a result of such a $W$ for a finite field $F_q$ of Proposition \ref{Proposition:4.9}.
	\end{proof}	
	\begin{coro} There exists $W=\{W^{(1)},W^{(2)},...,W^{(k)}\}$ such that $C$ is $W-\delta-\mu$ dual containing if and only if for all $j=1,2,...,k$ there exist vectors $g^{(1)},g^{(2)},...,g^{(s_j)}$ in $C_j$, where $s_j$ is the dimension of the Euclidean dual of $C_j$ such that for every $i=1,2,...,n$, $\delta^{(j)}(g_i^{(l)})\neq 0,\;\forall \;l=1,2,...,s_j.$
	\end{coro}
	
	\begin{proof} It follows from Theorem \ref{followings}(2) and the existence result for $W$, for $C$ to be dual containing over $F_q$ in Proposition \ref{5.8}.\end{proof} 
	
	\begin{coro} There exists $W=\{W^{(1)},W^{(2)},...,W^{(k)}\}$ such that $C$ is $W-\delta-\mu$ LCD code for any choice of $\delta\in\;\displaystyle\bigoplus_{j=1}^{k} Aut(F_q^{(j)})$ and $\mu=\{\mu^{(1)},\mu^{(2)},...,\mu^{(k)}\}$ with $\mu^{(j)}\in \; S_n$ if $q^{(j)}>k^{(j)}+1\;\forall\;j=1,2,...,k$, where $k^{(j)}$ is the dimension of $C_j$ over $F_q^j$.  \end{coro}
	
	\begin{proof} It follows from Theorem \ref{followings}(3) and the existence result for such a $W$ in case when the alphabet is $F_q$ in Proposition \ref{5.12}.\end{proof}
	\section{Conclusion and future scope}
	Classically, the most used inner products in coding theory have been the Euclidean and Hermitian products which have been explored somewhat exhaustively. A thorough study of duals, along with their interrelationships with the original codes with respect to these classical products with varied applications already exists in literature. Applications such as the CSS construction (using Euclidean dual) and the Hermitian construction (using Hermitian dual) of quantum codes signify their importance. To get specific and more varied applications it is useful to define more general products which provide more flexibility for desirable properties such as self duality, linear complementarity, etc. For example, $l-Galois$ product \cite{Fan} was used for quantum error correcting codes \cite{Liu}. In this paper, we have defined a product that generalizes Euclidean, Hermitian, $l-Galois$ and $\delta$ products. We have given a basic description, relations with Euclidean dual, characterizations of duals, existence conditions for self duality, complementary duals, among others, for $W-\delta-\mu$ product. We hope that the product finds applications in new quantum code constructions with the help of specific usage of the parameters $W,\delta$ and $\mu$ to get curated duals.

	The results obtained in this work indicate that the $W$--$\delta$--$\mu$ product constitutes a useful extension of existing inner products. The computational evidence demonstrates that the enlarged parameter space consisting of various choices of $W$ and the permutation $\mu$ leads to substantially more number of products than those available under the $\delta$-product \cite{Dinh}. hence, it increases the likelihood of obtaining LCD and self-orthogonal codes from the same generator matrix. Such flexibility may prove valuable in applications where multiple dual structures are desirable, allowing the selection of products best suited to particular algebraic or computational requirements.
	
	The generality of the proposed framework also opens several interesting directions for future research. The theory developed here can be extended to other important classes of linear codes, including \emph{cyclic, constacyclic, quasi-cyclic,} and \emph{quasi-twisted codes}, as well as to \emph{finite chain rings, Galois rings}, and other algebraic structures. From a cryptographic perspective, the large number of choices for $W$, together with the freedom of selecting suitable permutations, provides a significantly expanded space for constructing code families with multiple dual structures. This flexibility may be useful in the design of \emph{\textbf{code-based cryptographic}} schemes and \emph{\textbf{post-quantum cryptographic}} constructions, where a richer collection of cryptographic keys can enhance scheme diversity and potentially increase resistance to structural analysis. These directions constitute promising topics for further theoretical and computational investigation.
	
	The future scope of study with respect to the introuced product includes but not limited to the following:
	\begin{enumerate}
		\item Using different alphabets such as chain rings, Galois rings and more generally arbitrary finite rings with ambient space as modules.
		\item Exploring the existence conditions for $\delta$ and $\mu$ which guarantee duals with specific properties like self duality and complementarity.
		\item We have shown that the $W-\delta-\mu$ dual of constacyclic codes need not be constacylic, it is interesting to find the most general conditions on these parameters which preserve constacyclicity for duals.
		\item In case of semisimple rings, $W,\delta$ and $\mu$ are described componentwise. The case where $W,\delta,\mu$ are arbitrary may give novel results.
		\item Studying applications of the proposed product in code-based and post-quantum cryptography.
		
		\item Identifying optimal parameters $W$, $\delta$, and $\mu$ for constructing cryptographically secure code families.
		
		\section*{Acknowledgement} The second author is  grateful to the Council of Scientific and Industrial Research for funding this research, File number: 09/0001(20883)/2025-EMR-I.
		
	\end{enumerate}

\end{document}